\documentclass[letterpaper, 10 pt, conference]{ieeeconf}
\IEEEoverridecommandlockouts 
\usepackage{graphics} 
\usepackage{algorithm}
\usepackage[noend]{algpseudocode}
\usepackage{cite}
\usepackage{hyperref}
\usepackage{amsmath} 
\usepackage{amssymb}

\usepackage{amsthm}
\usepackage{mathtools}
\usepackage{xcolor}

\usepackage[capitalize]{cleveref}

\newtheoremstyle{theoremsty}
{5pt} 
{5pt} 
{} 
{} 
{\bfseries} 
{.} 
{ } 
{} 
\theoremstyle{theoremsty}

\newtheorem{proposition}{Proposition}
\newtheorem{lemma}{Lemma}

\theoremstyle{remark}
\newtheorem{remark}{Remark}

\def \A{\mathsf{A}}
\def \B{\mathsf{B}}

\newcommand{\Not}{%
  \relax{\sim}
}

\allowdisplaybreaks

\title{\LARGE \bf
Sequential Binary Hypothesis Testing with Competing Agents under Information Asymmetry
}
\author{Aneesh Raghavan, M. Umar B. Niazi, and Karl H. Johansson
\thanks{This work was partly supported by the Swedish Research Council's Distinguished Professor grant 2017-01078, Knut and Alice Wallenberg Foundation's Wallenberg Scholar grant, the Swedish Strategic Research Foundation's FuSS SUCCESS grant, and the European Union’s Horizon Research and Innovation Programme under Marie Skłodowska-Curie grant agreement No. 101062523.}
\thanks{The authors are with the Division of Decision and Control Systems, KTH Royal Institute of Technology, Stockholm SE-100 44, Sweden. Emails: {\tt\small aneesh@kth.se}, {\tt\small mubniazi@kth.se}, {\tt\small kallej@kth.se}}%
}

\begin{document}
\maketitle
\thispagestyle{empty}
\pagestyle{empty}

\begin{abstract}
This paper concerns sequential hypothesis testing in competitive multi-agent systems where agents exchange potentially manipulated information. Specifically, a two-agent scenario is studied where each agent aims to correctly infer the true state of nature while optimizing decision speed and accuracy. At each iteration, agents collect private observations, update their beliefs, and share (possibly corrupted) belief signals with their counterparts before deciding whether to stop and declare a state, or continue gathering more information. The analysis yields three main results: (1)~when agents share information strategically, the optimal signaling policy involves equal-probability randomization between truthful and inverted beliefs; (2)~agents maximize performance by relying solely on their own observations for belief updating while using received information only to anticipate their counterpart's stopping decision; and (3)~the agent reaching their confidence threshold first cause the other agent to achieve a higher conditional probability of error. Numerical simulations further demonstrate that agents with higher KL divergence in their conditional distributions gain competitive advantage. Furthermore, our results establish that information sharing---despite strategic manipulation---reduces overall system stopping time compared to non-interactive scenarios, which highlights the inherent value of communication even in this competitive setup.
\end{abstract}

\section{Introduction}

Decision-makers in multi-agent systems face uncertainty regarding both the environment and other agents' behavior. In this paper, we consider a scenario where agents interact with each other to improve their decisions while balancing two conflicting factors: the potential for complementary information received from their neighboring agents, and the risk that the information received might be misleading. 
In particular, we study a sequential binary hypothesis testing problem where two agents interact with each other to competitively infer the true state of nature. At each iteration, both agents collect local observations that are statistically related to the true state, update their beliefs derived from these observations, and share (possibly manipulated) beliefs with their counterparts. In this sequential framework, agents must make decisions over time based on the accumulated information, i.e., their own sets of observations and belief signals received from their counterparts. Using this information, both agents can choose from two possible actions at every iteration: stop and declare the state they believe to be true, or continue gathering more information. 

Despite the information exchange between agents, the problem remains competitive as each agent independently aims to infer the true state both quickly and accurately, while receiving possibly misleading information from the other agent. The information asymmetry arises because agents have access to only their local observations and knowledge of their respective marginal distributions.
Such sequential multi-agent decision-making problems under competition and information asymmetry are ubiquitous across various domains, spanning both economic and engineering applications.
For instance, in financial markets, traders make investment decisions based on private information and market signals, where strategic misrepresentation of beliefs can lead to market manipulation \cite{foster1994}. Similarly, in smart grids, distributed energy resources make generation and consumption decisions based on beliefs about grid conditions while receiving strategic signals from other resources competing for limited capacity or favorable pricing \cite{saad2012}. 
Other applications include decentralized supply chains \cite{guo2014}, resilient sensor networks \cite{ren2018binary}, and federated learning systems \cite{gupta2023}.


The applications above share a common structure: agents with access to only local information make sequential decisions while incorporating potentially misleading signals from others. The fundamental challenge lies in designing robust decision rules when information sources cannot be fully trusted, yet contain valuable insights that could improve decision quality. We address this challenge through a novel formulation that captures the essence of unreliable information exchange while maintaining analytical traceability.

\subsection{Related Literature} 

Decentralized hypothesis testing under information asymmetry is a classical problem in detection theory and distributed control. Tenney and Sandell \cite{tenney1981} established the collaborative framework for decentralized detection where sensors make local decisions based on independent observations, showing that optimal strategies typically involve likelihood ratio tests with coupled thresholds. Teneketzis and Varaiya \cite{teneketzis1982decentralized} extended these concepts to sequential problems, demonstrating how threshold-based stopping rules operate when agents must decide whether to continue gathering information.
Using a game-theoretic framework for multi-agent sequential hypothesis testing, \cite{kim2014} establishes that belief states serve as sufficient statistics for conditionally independent observations. However, these works assume that agents share true beliefs with their counterparts.

Strategic interactions in hypothesis testing under adversarial frameworks have also been studied recently. For instance, \cite{ren2018binary} analyzes the tradeoff between security and efficiency in binary hypothesis testing with Byzantine sensors, while \cite{vamvoudakis2014detection} and \cite{barni2014binary} formulate a similar problem as a zero-sum partial information game. 
However, our work considers environments where agents may act in their self-interest, but do not necessarily seek to inflict harm.

A related line of research addresses how agents can collectively learn despite limited information sharing. In this regard, a distributed hypothesis testing approach using a ``min-rule'' rather than belief-averaging has been proposed by \cite{mitra2020new}. Similarly, \cite{saritas2019} examines hypothesis testing as a signaling game with subjective priors or misaligned objectives. Another particularly relevant work is \cite{lalitha2018social}, which analyzes distributed hypothesis testing where network nodes communicate beliefs to neighbors and shows that beliefs on wrong hypotheses converge to zero exponentially fast. 
While these distributed approaches often focus on consensus-building protocols, our work specifically addresses the strategic tension between information sharing and withholding by modeling the possibility of deceptive signaling between agents and determining optimal response strategies.

\subsection{Our Contributions}


In this paper, we examine competitive sequential binary hypothesis testing in which agents may manipulate the information they share with their counterparts. The reason for deliberately manipulating the information is to delay the other agent in inferring the true state. Our work differs from existing literature in two key aspects: first, unlike adversarial models, we focus on strategic behavior where agents act in self-interest without malicious intent; second, unlike consensus-building protocols or collaborative models, we analyze the strategic nature of information exchange and develop optimal response strategies within a competitive decision-making framework.

Our results are summarized as follows. First, we show that when agents strategically share information, the optimal signaling policy involves randomizing between truthful and inverted beliefs (\cref{Proposition 2}), which results from effectively maximizing uncertainty while preserving information content. Second, we prove that it is optimal for agents to not incorporate information received from others in their local beliefs, relying solely on their own observations until the stopping time is detected (\cref{Proposition 3}). 
However, information received from others serves the agents to detect when their counterparts will stop. 
Finally, we establish that this dynamic creates an advantage for the agent achieving the confidence threshold before the other (\cref{prop:prob-error}). That is, the agent who reaches their confidence threshold first achieves a bounded conditional probability of error, while the other agent faces potentially less accurate inference.
Our information-theoretic interpretations through entropy offer an elegant explanation for these somewhat counterintuitive results.

\section{Problem Formulation}
This section describes the model of sequential multi-agent binary hypothesis testing and formulates the problem.

\begin{figure}[!t]
    \centering
    \includegraphics[width=0.7\linewidth]{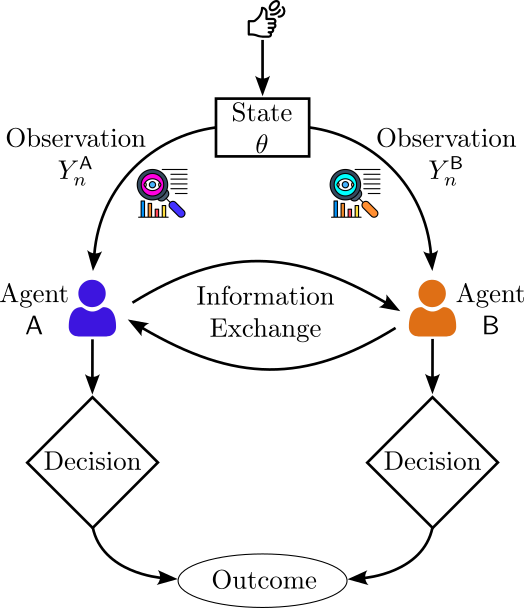}
    \caption{The two-agent binary hypothesis testing setup.}
    \label{fig:block-diagram-setup}
\end{figure}

\subsection{Setup}
Suppose there are two possible states of nature, $\theta^0$ and $\theta^1$, that are mutually exclusive. 
Consider two agents, $\A$ and $\B$, each collecting its local set of observations, which are statistically related to the true state of nature. Each agent has access to the joint distribution of their respective observations and the state. Given the observations, the objective of the agents is to find the true state $\theta$ competitively. This setup is illustrated in \cref{fig:block-diagram-setup} and formalized below:

\begin{itemize}
\item Both agents operate on the same time scale, i.e., they act simultaneously. Time is considered to be discrete and denoted by subscripts $n\in\mathbb N$, where $n$ will sometimes be also referred to as iteration. At every iteration $n$, there are multiple episodes where agents carry out multiple steps (e.g., interact with each other and make decisions).

\item The true state of nature $\theta\in\Theta \triangleq \{\theta^0,\theta^1\}$ remains fixed during a given experiment.

\item The observations collected by agent~$i$, where $i\in\{\A,\B\}$, are denoted by $Y^i_n\in S^i$, where $S^i$ is a finite set of real numbers or real vectors of finite dimension. 

\item The state of nature $\theta$ and the observations $\{Y^\A_n, Y^\B_n\}_{n\in\mathbb N}$ (and functions of these observations) are considered to be random variables on an abstract probability space $(\Omega, \mathcal{F}, \mathbb{P})$. That is, $\theta:\Omega\to \Theta$ and $Y_n^i:\Omega\to S^i$ are measurable functions, for $i\in\{\A,\B\}$. 

\item None of the agents have access to the joint distribution
$\mathbb{P}(\theta = \theta^k \cap \{Y^\A_t = y^\A_t\}^n_{t=1} \cap \{Y^\B_t = y^\B_t\}^n_{t=1} )$, where $k\in\{1,2\}$ and $(y^\A_t,y^\B_t)\in S^\A\times S^\B$. 
Each agent~$i$ knows only its marginal distribution, i.e., 
$\mathbb{P}(\theta = \theta^k \cap \{Y^i_t = y^i_t\}^n_{t=1})$. 
Moreover, both agents know a common prior $\mathbb P(\theta=\theta^k)$ and their respective conditional probabilities $\mathbb P(\{Y_t^i=y_t^i\}_{t=1}^n \mid \theta=\theta^k)$.
\end{itemize}

\begin{figure}[!t]
    \centering
    \includegraphics[width=\linewidth]{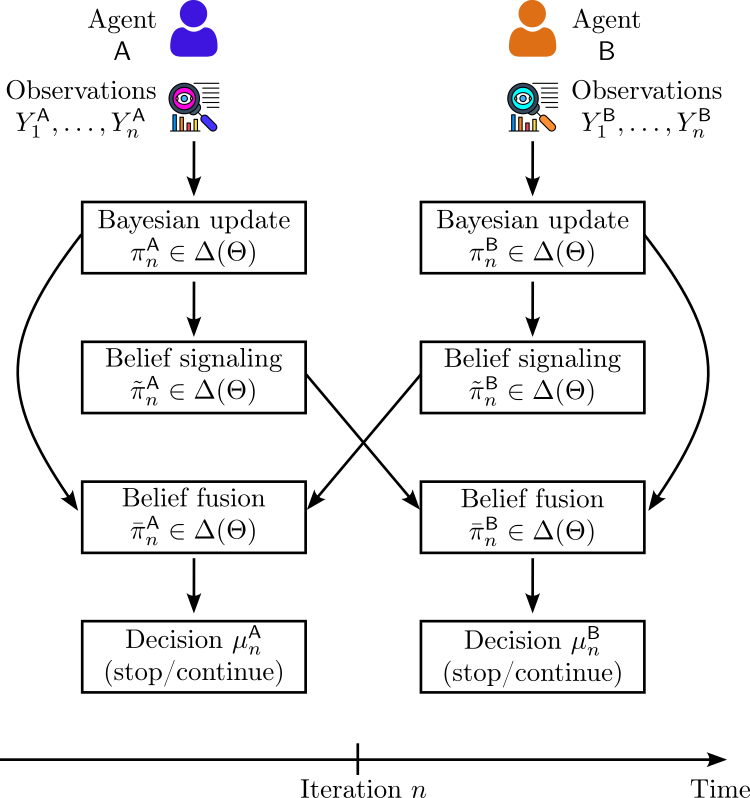}
    \caption{Steps executed by agents at each iteration of sequential binary hypothesis testing.}
    \label{fig:signaling-decision}
\end{figure}

\subsection{Belief Signaling and Decision Policies}
During iteration $n\in\mathbb N$, as illustrated in \cref{fig:signaling-decision}, each agent~$i\in\{\A,\B\}$ executes the following steps in sequence:
\begin{enumerate}
    \item Update the local belief $\pi^i_n \in \Delta(\Theta)$ based on the sample $Y^i_n$ collected. We consider a Bayesian update: 
     \begin{align*}
        \pi^i_n(\theta^k) &= \mathbb{P}(\theta = \theta^k \mid \{Y^i_t = y^i_t\}^n_{t=1}).
    \end{align*}
    \item Reveal a (possibly corrupted) belief $\tilde\pi^i_n$ to the other agent~$\Not i$, where\footnote{We use the notation $\Not i$ to denote the agent who is not $i$. Similarly, $\theta^{\Not k}$ means the state that is not $\theta^k$.}
        $\tilde\pi^i_n = \nu^i(\pi^i_n)$
    with $\nu^i:\Delta(\Theta)\to L^{1}(\Omega, \mathcal{F}, \mathbb{P})$ the \textit{belief signaling policy}.
    \item Modify the local belief $\bar\pi^i_n \in \Delta(\Theta)$ by incorporating the belief revealed by the other agent via \textit{belief fusion policy}. 
    \item Decide whether to stop and declare $D^i_n = \mu^i_n(\bar\pi^i_n)$ to be the most likely state of nature according to the modified belief $\bar\pi^i_n$ or continue with collecting the next observation $Y^i_{n+1}$. Here, $\mu^i_n: \Delta(\Theta) \to \Theta$ is the \textit{decision policy} of agent~$i$.
\end{enumerate}

\emph{Belief signaling policy.}
We consider the following structure of the belief signaling policy:
\begin{equation}
\label{eq:signaling}
    \nu^i ( \pi)  =
    \left\{\begin{array}{ll}
      \pi ,   & \text{with probability} \; \alpha^i  \\
      1 - \pi ,    &  \text{with probability} \; 1- \alpha^i
    \end{array}\right.
\end{equation}
where $0 \leq \alpha^i \leq 1$ is a design parameter corresponding to how much $i$ should be truthful. In particular, agent~$i$ flips a biased coin (probability of head = $\alpha^{i}$, probability of tail = $1 - \alpha^{i}$) to decide what signal should be sent. If the outcome is head, the agent sends their true belief ($\tilde\pi^{i}_n = \pi^{i}_n$), else the agent sends the inverted belief ($\tilde\pi^{i}_n = 1- \pi^{i}_n$) to agent~$\Not i$. The outcome of the coin toss for agent~$i$ is denoted by $Z^{i}_{n}$, where $Z^{i}_{n} : \Omega \to \{0,1\}$ is a measurable function. Thus, $\{Z^i_n$\} is a  Bernoulli process and  for agent~$i$, the process $\{Z^i_n$\} is assumed to be independent of the process $\{Y^i_n$\}.  The belief distribution of $\theta$ due to $\tilde\pi^i_n$ is denoted by 
\[
\begin{bmatrix}
\mathbb{P}(\theta = \theta^1 \mid  \{Y^i_t = y^i_t\}^n_{t=1})\\
\mathbb{P}(\theta = \theta^0 \mid  \{Y^i_t = y^i_t\}^n_{t=1})
\end{bmatrix}_{\Big|_{i \to \Not i }}. 
\]

\emph{Belief fusion policy.}
For each agent~$i$, we consider the belief fusion policy to be a convex combination of its own belief and the belief revealed by the other agent:
\begin{multline}
\label{eq:belief-fusion}
    \bar\pi^i_n = w^i_n \pi^i_n \\  
    +(1-w^i_n) \begin{bmatrix}
    \mathbb{P}(\theta = \theta^1 \mid  \{Y^{\Not i}_t = y^{\Not i}_t\}^n_{t=1})\\
    \mathbb{P}(\theta = \theta^0 \mid  \{Y^{\Not i}_t = y^{\Not i}_t\}^n_{t=1})
    \end{bmatrix}_{\Big|_{{\Not i} \to  i }}
\end{multline}
where $0 \leq w^i_n\leq 1$ is a design parameter corresponding to the weight $i$ puts on its own observations as compared to the belief revealed by the other agent. Note that the aim of the belief fusion policy \eqref{eq:belief-fusion} is to efficiently incorporate the belief revealed by the other agent to infer the true state of nature as fast as possible.

\emph{Decision policy.}
Following the literature on stopping time problems \cite{shiryaev2007optimal}, \cite{tartakovsky2014sequential}, we consider the decision policy $\mu^i_n(\bar\pi^i_n)$ to be a threshold policy: 
\begin{equation}
\label{Eq: Decision-Policy}
    \mu^i_n(\bar\pi^i_n)  = 
    \left\{\arraycolsep=2pt\begin{array}{ll}
    \text{stop}, ~ D^i_n = \theta^1; & \text{if} ~ \bar\pi^i_n(\theta^0) \leq T^i_{\theta^0} \\
    \text{stop}, ~ D^i_n = \theta^0;  &  \text{if} ~ \bar\pi^i_n(\theta^1) \leq T^i_{\theta^1} \\
    \text{continue}, ~ D^i_n = \varnothing; & \text{otherwise}
\end{array}\right.
\end{equation}
where $T^i_{\theta^0}, T^i_{\theta^1} \in[0,1]$ are threshold parameters to be designed. The decision policy \eqref{Eq: Decision-Policy} aims to infer the true state of nature quickly and achieve the desired level of accuracy.

\subsection{Problem Definition}
Following stochastic control formulations of the hypothesis testing problem as in \cite{kumarvaraiya2015, tenney1981, teneketzis1982decentralized}, we consider the cost function for agent $i$ to be
\begin{multline*}
    C^i(D^i_n) = \mathbb{E}_{\mathbb{P}}\Big[nc^i \mathbf{1}_{D^i_n = \varnothing} + \\ \hat c^i \mathbf{1}_{D^i_n \neq \varnothing, D^i_n \neq \theta}  ~\Big|~ \{Y^{i}_{t}, \tilde{\pi}^{\Not i}_{t}\}^n_{t=1} \Big]
\end{multline*}
where $c^i>0$ is the cost associated with continuing and collecting the next sample $Y^i_{n+1}$, and $\hat c^i>0$ is the cost associated with stopping when the declared state $D^i_n$ is wrong (i.e., $D^i_n\neq\varnothing$ and $D^i_n\neq \theta$). The optimization problem for each agent at iteration $n$ is thus defined as:
\begin{align}
    \underset{w^i_n, T^i_{\theta^0}, T^i_{\theta^1}}{\min } C^i(D^i_n) \label{Problem 1}
\end{align}
where $D^i_n \coloneqq D^i_n(w^i_n, T^i_{\theta^0}, T^i_{\theta^1})$.

Since the joint distribution between $\theta$, $\{Y^\A_t\}_{t=1}^n$, and $\{Y^\B_t\}^n_{t=1}$ is unknown to agent~$i$, the joint distribution between $\theta$ and $\{Y^{i}_{t} ,\tilde{\pi}^{\Not i}_{t}\}^n_{t=1}$ is  also unknown to the agent. Therefore, finding the conditional probability $\mathbb{E}_{\mathbb{P}} \big[ \mathbf{1}_{\theta = \theta^1} \mid \{Y^{i}_{t}, \tilde{\pi}^{\Not i}_{t}\}^n_{t=1} \big]$ using Bayes theorem is not possible. In order to analyze the optimization problem \eqref{Problem 1}, define
\begin{align*}
    &\mathbb{E}_{\mathbb{P}^i} \Big[ \mathbf{1}_{\theta = \theta^1}   \Big| \{Y^{i}_{t}, \tilde{\pi}^{\Not i}_{t}\}^n_{t=1}   \Big] \triangleq  \bar\pi^i_n (\theta^1), \; \text{and} \;    \\
    &\mathbb{E}_{\mathbb{P}^i} \Big[ \mathbf{1}_{\theta = \theta^0}   \Big| \{Y^{i}_{t}, \tilde{\pi}^{\Not i}_{t}\}^n_{t=1}   \Big] \triangleq  \bar\pi^i_n (\theta^0) = 1-  \bar\pi^i_n (\theta^1).
\end{align*}
Then, the cost function of \eqref{Problem 1} can be simplified as 
\begin{multline*}
   C^i(D^i_n) =  nc^i\mathbf{1}_{D^i_n = \varnothing} \\ + \hat{c}^i \Big[\mathbf{1}_{D^i_n = \theta^1}\bar\pi^i_n (\theta^0) + \mathbf{1}_{D^i_n = \theta^0}\bar\pi^i_n (\theta^1) \Big].
\end{multline*}
Note that since
\begin{multline}
    \mathbf{1}_{D^i_n = \theta^1}\bar\pi^i_n (\theta^0) + \mathbf{1}_{D^i_n = \theta^0}\bar\pi^i_n (\theta^1) \\ = \mathbb{P}(\mu^{i}_{n}(\bar\pi^{i}_{n}) \neq \theta \mid \{Y^{i}_{t}, \tilde{\pi}^{\Not i}_{t}\}^n_{t=1}  ) \label{Eq: Probability of error}
\end{multline}
it follows that
\begin{align*}
    C^i(D^i_n) =  nc^i&\mathbf{1}_{D^i_n =\varnothing}
    + \\
    &\hat{c}^i \mathbb{P}(\mu^{i}_{n}(\bar\pi^{i}_{n}) \neq \theta \mid \{Y^{i}_{t}, \tilde{\pi}^{\Not i}_{t}\}^n_{t=1}). 
\end{align*}

Let $\tau^i$ denote the random stopping time at which agent~$i$ stops. Then, 
\begin{align*}
    C^i(D^i_{\tau^i}) =  c^i \tau^i + \hat{c}^i \mathbb{P}(\mu^{i}_{\tau^i}(\bar\pi^{i}_{\tau^i}) \neq \theta \mid \{Y^{i}_{t}, \tilde{\pi}^{\Not i}_{t}\}^n_{t=1} ). 
\end{align*}
Thus, the cost is a tradeoff between the number of samples $\tau^i$ and the conditional probability of error $\mathbb{P}(\mu^{i}_{\tau^i}(\bar\pi^{i}_{\tau^i}) \neq \theta \mid \{Y^{i}_{t}, \tilde{\pi}^{\Not i}_{t}\}^n_{t=1})$. Different values of $c^i$ and $\hat{c}^i$, more specifically the ratio 
$\hat{c}^i/c^i$,
lead to different probabilities of error for a given stopping time and vice-versa. When the ratio is high, then the emphasis is on the probability of error, and the agent collects samples until the probability of error is sufficiently small. When the ratio is very low, the emphasis is on the stopping time, and the agent collects fewer samples at the expense of a high probability of error. Thus, problem \eqref{Problem 1} can be viewed in terms of a Lagrangian formulation:
\begin{subequations}
\label{Problem 2}
\begin{align}
    \underset{\{w^i_{n}\}^{\tau^i}_{n=1}, T^i_{\theta^0}, T^i_{\theta^1} }{\min}\quad &\tau^i  \label{Problem 2: Objective} \\
    \text{subject to}\quad & \mathbb{P}(\mu^{i}_{\tau^i}(\bar\pi^{i}_{\tau^i}) \neq \theta \mid \{Y^i_t\}^{\tau^i}_{t=1}) \leq \beta^i \label{Problem 2: Constraint}
\end{align}
\end{subequations}
where different confidence thresholds $\beta^i$'s correspond to different stopping thresholds $T^i_{\theta^0}, T^i_{\theta^1}$.

\emph{Signaling problem.}
In the competitive setup, the agent who wins is the one achieving the confidence threshold before the other agent. To ensure fairness among the agents, we impose the condition that $\beta^1 = \beta^2 = \beta$. In other words, the agent who stops first while ensuring that the probability of error conditioned on its observations is less than equal to $\beta$ wins the competition. For any belief fusion policy and thresholds of agent~${\Not i}$, agent~$i$ aims to prevent $\Not i$ from winning the competition by optimizing over the belief signaling policy. Since the belief signaling policy \eqref{eq:signaling} is parameterized by $\alpha^i$, the objective of agent~$i$ is
\begin{align}
\underset{\alpha^i }{\max}\quad & \mathbb{P}(\mu^{\Not i}_{n}(\bar\pi^{\Not i}_{n}) \neq \theta \mid \{Y^{\Not i}_{t}, \tilde{\pi}^{i}_{t}\}^n_{t=1} ). \label{Problem 3}
\end{align}
That is, agent~$i$ chooses the truthfulness probability $\alpha^i$ such that it maximizes $\Not i$'s conditional probability of error. Although solving \eqref{Problem 3} is not possible for agent~$i$ because it doesn't know the parameters of $\Not i$'s belief fusion and decision policies, one can indirectly characterize the optimal value of $\alpha^i$. Note that the aim of the signaling policy according to \eqref{Problem 3} is to not reveal a lot of information so that the other agent gains sufficient information to correctly and rapidly infer the true state, but also to not corrupt it too much so that the signal becomes predictable enough and can be simply inverted. 

\emph{Competitive stopping problem.}
By incorporating the signaling problem \eqref{Problem 3} into the stopping problem \eqref{Problem 2}, we obtain a competitive stopping problem as
\begin{subequations}
\label{Problem 4}
\begin{align}
\underset{\{w^i_{n}\}^{\tau^i}_{n=1}, T^i_{\theta^0}, T^i_{\theta^1} }{\min} &\tau^i  \label{Problem 4: Objective} \\
\text{subject to}~ &  \underset{\alpha^{\Not i}}{\max}~ \mathbb{P}(\mu^{i}_{\tau^i}(\bar\pi^{i}_{\tau^i}) \neq \theta \mid \{Y^i_t\}^{\tau^i}_{t=1}) \leq \beta^i. \label{Problem 4: Constraint}
\end{align}
\end{subequations}
That is, agent~$i$ minimizes its stopping time subject to achieving conditional probability of error irrespective of signals it receives from $\Not i$.
Note that the belief fusion policy is parameterized by $w^i_n$ and the decision policy is parametrized by $T^i_{\theta^0}$ and $T^i_{\theta^1}$. That is, the problem boils down to designing parameters $w^i$, $T^i_{\theta^0}$, and $T^i_{\theta^1}$ of the belief fusion policy \eqref{eq:belief-fusion} and the decision policy \eqref{Eq: Decision-Policy} such that the above objective is achieved.

\section{Analysis}
This section presents our key results on sequential binary hypothesis testing with competing agents. We characterize the optimal signaling and belief fusion policies. We also derive the optimal stopping thresholds and the conditional probability of errors of the agents.

\subsection{Optimal Signaling}

We explore how agents should optimally manipulate their beliefs when communicating with their counterparts. Considering the belief signaling policy \eqref{eq:signaling}, where each agent has a ``truthfulness'' parameter $\alpha^i$, we observe that setting this truthfulness parameter to exactly 0.5 creates an optimal signaling strategy. This means the agent should randomly choose between sending their true belief or its opposite with equal probability, regardless of the other agent's stopping thresholds or belief fusion parameters. 

\begin{proposition}\label{Proposition 2}
Consider the belief signaling policy \eqref{eq:signaling} parametrized by $0\leq \alpha^i\leq 1$. Then, for any belief fusion parameter sequence, $\{w^{\Not i}_{n}\}$ and thresholds $T^{\Not i}_{U}, T^{\Not i}_{L}$ of agent~$\Not i$, $\alpha^i = 0.5$ is optimal with respect to \eqref{Problem 3}. 
\end{proposition}
\begin{proof}
Without loss of generality, we analyze the problem from agent~$\A$'s perspective.
The belief distribution of $\theta$ due to $\tilde\pi^\A_n$ can be expressed as follows. 
\begin{align*}
&\mathbb{P}(\theta = \theta^k \mid  \{Y^\A_t = y^\A_t\}^n_{t=1}) \Big|_{\A \to \B}  \\
& \qquad \triangleq \mathbb{P}(Z^\A_n = H  \cap  \theta = \theta^k \mid  \{Y^\A_t = y^\A_t\}^n_{t=1} ) \\
& \qquad \qquad \qquad + \mathbb{P}(Z^\A_n = T  \cap  \theta = \theta^{\Not k } \mid  \{Y^\A_t = y^\A_t\}^n_{t=1} ) \\
& \qquad =\mathbb{P}(Z^\A_n = H)\mathbb{P}(\theta = \theta^k \mid  \{Y^\A_t = y^\A_t\}^n_{t=1} ) \\
& \qquad \qquad \qquad + \mathbb{P}(Z^\A_n = T)\mathbb{P}(\theta = \theta^{\Not k } \mid  \{Y^\A_t = y^\A_t\}^n_{t=1} )\\
& \qquad = \alpha^\A\pi^\A_n(\theta^k) + (1-\alpha^\A)\pi^\A_n(\theta^{\Not k}),
\end{align*}
where random variable $Z^\A_n$ denotes the outcome of the coin toss (head $=H$, tail $=T$) for agent~$\A$ to decide whether to flip the belief or not. Stacking the beliefs on states gives
\begin{align*}
\begin{bmatrix}
\mathbb{P}(\theta = \theta^1 \mid  \{Y^\A_t = y^\A_t\}^n_{t=1})\\
\mathbb{P}(\theta = \theta^0 \mid  \{Y^\A_t = y^\A_t\}^n_{t=1})
\end{bmatrix}_{\Big|_{\A \to \B}}  \hspace{-7pt}= 
\begin{bmatrix}
\alpha^\A &1-\alpha^\A \\
1-\alpha^\A &\alpha^\A
\end{bmatrix} \pi^\A_n.
\end{align*}
From the definition of $\bar{\pi}^\B_n$, we note that it is affine in $w^\B_n$, $\pi^\B_{n}$, and $\tilde{\pi}^\A_n$ \textit{individually}, i.e. when one of the three is varied and the other two are kept fixed, $\bar{\pi}^\B_n$ is linear in the quantity being varied. Due to this affine structure, given any $w^\B_n, \pi^\B_{n}$, and thresholds $T^\B_{\theta^0}, T^\B_{\theta^1}$, agent $\A$ aims to choose $\alpha^\A$ such that the probability of error of $\B$ due to $\tilde\pi^\A_n$ is maximized. That is, agent $\A$ could assume $w^\B_n=0$ and suitably design $\alpha^\A$. If agent~$\B$ is to make its decision based on $\bar\pi^\B_n = \tilde\pi^\A_n$, then it decides $D^\B_n = \theta^k$ if and only if 
\begin{multline*}
\alpha^\A\pi^\A_n(\theta^{\Not k}) + (1-\alpha^\A)\pi^\A_n(\theta^k) \\ \leq \alpha^\A\pi^\A_n(\theta^k) + (1-\alpha^\A)\pi^\A_n(\theta^{\Not k}).
\end{multline*}
Thus, the objective of the agent~$\A$ is to maximize over $\alpha^\A$ the minimum of the two quantities (the right-hand side and the left-hand side) in the above inequality. Without loss of generality, suppose $D^\A_n = \theta^k$, i.e., $\pi^\A_n(\theta^{\Not k}) \leq \beta$. Then,
\begin{itemize}
    \item For $0 \leq \alpha^\A < 1/2$, the minimum is 
    \[\alpha^\A\pi^\A_n(\theta^k) + (1-\alpha^\A)\pi^\A_n(\theta^{\Not k}).\]
    \item For $1/2 < \alpha^\A \leq 1$, the minimum is
    \[\alpha^\A\pi^\A_n(\theta^{\Not k})+ (1-\alpha^\A)\pi^\A_n(\theta^k).\]
    \item For $\alpha^\A=1/2$, the both quantities are equal.
\end{itemize}
This can be verified by using $\pi^\A_n(\theta^{\Not k}) = 1- \pi^\A_n(\theta^k)$ and expanding the two terms. Thus, the minimum of the two quantities
\begin{itemize}
    \itemsep3pt
    \item $\mathbb{P}(\theta = \theta^1 \mid  \{Y^\A_t = y^\A_t\}^n_{t=1}) \Big|_{\A \to \B}$
    \item $\mathbb{P}(\theta = \theta^0 \mid  \{Y^\A_t = y^\A_t\}^n_{t=1}) \Big|_{\A \to \B}$
\end{itemize}
is concave in $\alpha^\A$ with maximum achieving at $\alpha^\A= 0.5$.
\end{proof}

\begin{figure}
    \centering
    \includegraphics[width = 0.45\textwidth]{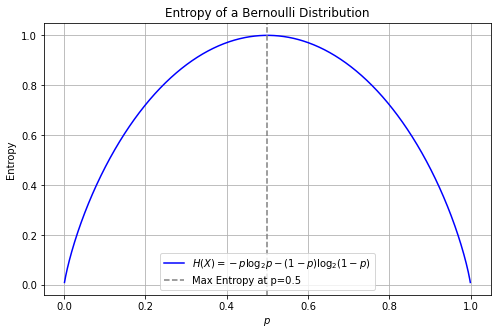}
    \caption{Entropy of the Bernoulli Distribution.}
    \label{fig:entropy}
\end{figure}

\emph{Entropy-based interpretation.}
    The intuition behind \cref{Proposition 2} connects to information entropy. By choosing a 50/50 split between truth and inversion, the agent maximizes uncertainty in their signals while still conveying useful information. This approach effectively hides the agent's true decision while ensuring the communicated belief maintains the same level of uncertainty as the original belief. Since the entropy of a belief is identical to its inverse, this strategy preserves information content while maintaining strategic ambiguity about the agent's actual position.
    
    To elucidate, recall the entropy of a random variable with Bernoulli distribution $p$:
    \begin{align*}
        H(p) = - p \log_{2}p -(1-p)\log_2(1-p).
    \end{align*}
    As seen in \cref{fig:entropy}, the entropy function possesses two geometric properties: (i)~symmetry about ${p=0.5}$ and (ii)~concavity in $p$. 
    The belief signaling policy \eqref{eq:signaling} utilizes this symmetry property. Since the entropy of $\pi^{i}_n$ or $1- \pi^i_n$ is the same, an agent can communicate to the other agent by transmitting either of them that it has achieved the desired level of uncertainty in the decision while keeping the true decision hidden. For any two independent random variables $Y$ and $Z$, it holds that $H(Y,Z) = H(Y) + H(Z \mid Y ) = H(Y) + H(Z)$. When an agent is trying to choose $\alpha^i$ such that the entropy of the distribution of the random vector $\tilde\pi^i_n$ gets maximized, it follows from the previous relation that the agent has to maximize $H(\alpha)$ because $H(\tilde\pi^i_n) = H(\alpha) + H(\pi^{i}_n)$. Since the entropy of the Bernoulli distribution gets maximized at $0.5$, it is optimal for agent~$i$ to choose $\alpha^i=0.5$. 
    \hfill $\diamond$

\begin{remark}\label{Remark 1}
The above result states that by choosing $\alpha^A =0.5$, we get 
\begin{align*}
\begin{bmatrix}
\mathbb{P}(\theta = \theta^1 \mid  \{Y^\A_t = y^\A_t\}^n_{t=1})\\
\mathbb{P}(\theta = \theta^0 \mid  \{Y^\A_t = y^\A_t\}^n_{t=1})
\end{bmatrix}_{\Big|_{\A \to \B}}  \hspace{-7pt}= 
\begin{bmatrix}
0.5 \\
0.5 
\end{bmatrix}
\end{align*}
which implies that, 
\begin{equation}
\label{eq:bar-pi}
\bar{\pi}^\B_n = w^\B_n \pi^\B_n  + (1-w^\B_n)\begin{bmatrix}
0.5 \\
0.5 
\end{bmatrix}
\end{equation}
Hence, for any  $0\leq w^\B_n <1$, the probability of error based on decision $\bar{\pi}^\B_n$ is greater than or equal to the probability of error based on decision  $\pi^\B_n$.
\hfill $\lrcorner$
\end{remark}

\subsection{Optimal Belief Fusion and Stopping Time}

We analyze how agents should optimally combine their own beliefs with information received from others, and when they should stop the process to make a final decision about the true state of nature. For the analysis, we consider belief fusion policy \eqref{eq:belief-fusion} where each agent uses a weighted combination of their local observations and the beliefs shared by other agents, alongside a threshold-based policy \eqref{Eq: Decision-Policy} to determine when to stop and declare the state.

As a preliminary to our key result, consider a single-agent setup. Notice that the decision of an agent in this setup is only a function of their Bayesian update $\pi^i_n$. Given a threshold $\beta$, the agent can verify at every iteration~$n$ whether $\pi^i_n(\theta^{\Not k}) \leq \beta$, for some $k\in\{0,1\}$. If this inequality is satisfied, then the agent stops and announces its decision as $D^i_n = \theta^k$. Otherwise, the agent collects the next observation, and the process continues.  
\begin{lemma}\label{Proposition 1}
For almost all sample paths (except a set of $\mathbb{P}$ measure 0), any desired conditional probability of error is achieved by the agents individually. That is, for every $\beta > 0$, there exists $N^i_{\beta}(\omega) \in\mathbb N$ such that either $\pi^{i}_{n}(\theta^0)(\omega) < \beta$ or  $\pi^{i}_{n}(\theta^1)(\omega) < \beta$, $\forall n > N^i_{\beta}$, for almost all $\omega \in \Omega$.
\end{lemma}
The proof of \cref{Proposition 1} follows from \textit{Doob's Martingale Theorem} (see \cite{billingsley2012probability, koralov2007theory}).

In our multi-agent setup, the agent achieving the desired conditional probability of error at iteration $n$
\[
\mathbb{P}(\mu^{i}_n(\bar\pi^{i}_n) \neq \theta \mid \{Y^i_t\}^n_{t=1}) \leq \beta^i
\]
will initiate the stopping, i.e., $\tau^i=n$. In the following, we show that all agents will ultimately stop the process at the same time, regardless of their local observation processes. Moreover, expanding on \cref{Remark 1}, we show that the optimal strategy for agents is to rely exclusively on their own observations when making decisions until the stopping time. 
This means agents should essentially ignore the information shared by their counterparts until they detect that one of them has achieved the desired conditional probability of error. 

\begin{proposition}\label{Proposition 3}
For almost all sample paths, both agents stop at the same time, i.e., $\tau^\A(\omega)=\tau^\B(\omega)=\tau(\omega) $. 
Moreover, for all agents $i\in\{\A,\B\}$, the optimal fusion parameter sequence with respect to \eqref{Problem 4} is given by $w^i_n= 1$, for $n \leq \tau$. For agent~$i$ initiating the stopping, the decision policy at $\tau$ is given by 
\begin{align}
\mu^{i}_{\tau}(\bar\pi^{i}_{\tau})  = 
\begin{cases}
D^{i}_{\tau} = \theta^1, & \text{if} \; \bar\pi^{i}_{\tau}(\theta^0) \leq \beta  \\
D^{i}_{\tau} = \theta^0,  &  \text{if} \; \bar\pi^{i}_{\tau}(\theta^1) \leq \beta.
\end{cases}      \label{Eq: Decision Policy Stopping Agent}
\end{align}
That is, $T^i_{\theta^0} = T^i_{\theta^1} =\beta$ is optimal for $i$ with respect to \eqref{Problem 4}.
The policy of the other agent~$\Not i$ at $\tau$ is
\begin{align}
\mu^{\Not i}_{\tau}(\bar\pi^{\Not i}_{\tau})  = 
\begin{cases}
\; D^{\Not i}_{\tau} = \theta^1 & \text{if} \; \bar\pi^{\Not i}_{\tau}(\theta^0) \leq \bar\pi^{\Not i}_{\tau}(\theta^1)  \\
\; D^{\Not i}_{\tau} = \theta^0  &  \text{if} \; \bar\pi^{\Not i}_{\tau}(\theta^1) \leq \bar\pi^{\Not i}_{\tau}(\theta^0). \label{Eq: Decision Policy Alternate Agent}
\end{cases}   
\end{align}
\end{proposition}
\begin{proof}
From the definition of the decision policy $\mu^{i}_{n}$ in \eqref{Eq: Decision-Policy}, the conditional probability of error \eqref{Eq: Probability of error}), and the constraint \eqref{Problem 4: Constraint} of problem \eqref{Problem 4}, it follows that the optimal thresholds are $T^\A_{\theta^0}= T^\A_{\theta^1} = T^\B_{\theta^0} = T^\B_{\theta^1} = \beta$. From \cref{Proposition 1}, it follows that for any sample path (or except for a set of $\mathbb{P}$ measure $0$), the constraint in Problem~\eqref{Problem 2} is achieved at a random time by both agents. For any belief parameter sequence $\{w^{i}_n\}$, if $\pi^\A_n(\theta^{ k}) > \beta$ and $\pi^\B_n(\theta^{k}) > \beta$, for all $k\in\{0,1\}$, then it follows from \eqref{eq:bar-pi} that the probabilities of error associated with decisions $\mu^\A_n(\bar\pi^\A_n)$ and $\mu^\B_n(\bar\pi^\B_n)$ are greater than $\beta$. Thus, the only way an agent satisfies the constraint \eqref{Problem 4: Constraint} is if its local Bayesian update $\pi^i_n$ achieves the same. 

The random times at which the agents satisfy their constraints could be different. For a given sample path, suppose without loss of generality that agent~$\A$ satisfies the constraint \eqref{Problem 4: Constraint} and immediately stops to achieve the minimum time. In this case, $\pi^\A_n(\theta^{ k}) \leq \beta$. But suppose $\pi^\B_n(\theta^{k}) > \beta$, for all $k\in\{0,1\}$. That is, agent~$\A$ initiates the stopping of the process. However, due to the belief signaling policy \eqref{eq:signaling}, agent~$\B$ receives $\tilde\pi^\A_n$ and realizes either $\tilde\pi^\A_n(\theta^{\Not k})\leq \beta$ (in case $\A$ lies) or $\tilde\pi^\A_n(\theta^{k})\leq \beta$ (in case $\A$ is truthful). Whatever the case, agent~$\B$ knows that agent~$\A$ has achieved the constraint~\eqref{Problem 4: Constraint} and is going to stop at this iteration. Due to the competitive setup under consideration, irrespective of whether $\B$ has achieved the constraint~\eqref{Problem 4: Constraint}, it also stops. Thus, the equality of the stopping times follows. 

Suppose $\pi^\A_n(\theta^{\Not k}) \leq \beta$ at iteration $n$. Disregarding $\tilde{\pi}^\B_n$, the agent can let $\bar\pi^\A_n = \pi^\A_n$, stop and announce its decision as $D^\A_n = \theta^k$. The decision policy for agent~$\A$ at the stopping time $\tau$ could be expressed by \eqref{Eq: Decision Policy Stopping Agent}. Recall from the previous paragraph that agent~$\B$ is aware that agent~$\A$ is going to stop at iteration $n$ because either $\tilde\pi^\A_n(\theta^{k}) \leq \beta$ or $\tilde\pi^\A_n(\theta^{\Not k}) \leq \beta$; however, it does not know the decision $D^\A_n$ of $\A$. Even though agent~$\B$ stops at iteration $n$, it is optimal for agent~$\B$ to make its decision based on $\pi^\B_n$ (cf. \cref{Remark 1}), i.e., $w^\B_{\tau} =1$. The decision policy for agent~$\B$ is given by \eqref{Eq: Decision Policy Alternate Agent}, i.e., the thresholds at stopping are updated and the constraint in problem~\eqref{Problem 4} is not achieved for the agent. 

Finally, suppose $\pi^\A_n(\theta^{ k}) > \beta$, for all $k\in\{0,1\}$, at iteration $n$. Irrespective of receiving the true belief or corrupted belief from agent~$\A$, agent~$\B$ knows that agent~$\A$ is not going to stop at the current iteration because both $\tilde\pi^\A_n(\theta^{k}) > \beta$ and $\tilde\pi^\A_n(\theta^{\Not k}) > \beta$ according to \eqref{eq:signaling}. Hence, agent~$\B$ could ignore agent~$\A$'s corrupted belief and base its decision only on its local belief, i.e. $w^\B_n=1$, $\forall n < \tau$. 
\end{proof}

The above proposition shows that even though $i$ may initiate the stopping, $\Not i$ can detect it through $i$'s revealed belief $\tilde\pi^i_n$ and will also stop at the same time. However, we show in the next section that $\Not i$, the agent who did not initiate the stopping, may lose the competition in the sense that it incurs a higher probability of error.
Another observation in the proposition is that even when agents receive information from their counterparts, none of the agents find it optimal to incorporate it into their local beliefs. However, they leverage this shared information to anticipate when a counterpart will stop, and only then do they incorporate it into their beliefs.
This counterintuitive result stems from the information-theoretic properties of belief fusion. 
When one agent shares manipulated information to maximize uncertainty (as established in the previous section), the receiving agent's best response is to disregard this deliberately ambiguous signal entirely. 

\emph{Entropy-based interpretation.} 
    The framework can be understood through the lens of entropy. Note that the belief fusion policy utilizes the concavity of the entropy function. From the concavity of $H (\cdot)$ it follows that
    \begin{multline*}
        H(\bar{\pi}^i_n) \geq w^{i}_{n}H(\pi^{i}_n) \\
        +(1-w^{i}_{n})H\Bigg(\begin{bmatrix}
        \mathbb{P}(\theta = \theta^1 \mid  \{Y^{\Not i}_t = y^{\Not i}_t\}^n_{t=1})\\
        \mathbb{P}(\theta = \theta^0 \mid  \{Y^{\Not i}_t = y^{\Not i}_t\}^n_{t=1})
        \end{bmatrix}_{\Big|_{{\Not i} \to  i }}\Bigg).
    \end{multline*}
    Therefore, for an agent aiming to choose the belief fusion parameter for which $H(\bar{\pi}^i_n)$ is minimized, it has to choose $w^{i}_{n} = 1 $ or $w^{i}_{n}=0$ depending on which of the terms they multiply has a lower value. Since agents are incentivized to maximize the entropy of their shared beliefs, i.e., agent $\Not i$ chooses $\alpha^i$ to maximize the term that multiples $w^{i}_{n}=0$, recipients minimize their uncertainty by trusting only their own observations, i.e., agent $i$ inevitably chooses $w^{i}_n=1$ for $n<\tau^i$.
    The optimal confidence thresholds then balance the cost of continuing to collect samples against the potential cost of making an incorrect decision.
    \hfill $\diamond$

\subsection{Conditional Probability of Error}

Previously, we established that both agents stop at the same time. However, what’s the benefit to the agent initiating the stopping? This section analyzes the accuracy of agents' decisions within our sequential binary hypothesis testing framework. We derive the conditional probability of error for each agent at the time they stop collecting information and make their final decision. 

\begin{proposition}
\label{prop:prob-error}
For almost all sample paths (except a set of $\mathbb{P}$ measure $0$), the conditional probability of error of the agent initiating the stopping is less than or equal to $\beta$, while for the other agent is less than or equal to $0.5$. 
\end{proposition}
\begin{proof}
For a given sample path, if both agents achieve the confidence threshold at the same iteration, then both of them achieve the desired accuracy.
Otherwise, only the agent that initiates the stopping achieves the desired accuracy (i.e., the conditional probability of error less than or equal to $\beta$). Agent~$\Not i$ by receiving $\tilde\pi^{i}_n$ detects that $i$ will stop at iteration~$n$, so it follows decision policy \eqref{Eq: Decision Policy Alternate Agent}, under which the confidence threshold for $\Not i$ is $\min(\pi^{\Not i}_{\tau}(\theta^0),\pi^{\Not i}_{\tau}(\theta^1)) \leq 0.5$. 
\end{proof}

The above proposition shows an important competitive dynamic between agents: the agent who initiates the stopping process (by reaching their confidence threshold first) achieves a conditional probability of error that is bounded by the predefined parameter $\beta$. Meanwhile, the other agent, who must make their decision under insufficient information, can only guarantee that their conditional probability of error is bounded by $0.5$.
This result highlights the strategic advantage gained by the agent who reaches confidence in their decision first. By initiating the stopping process, the winning agent secures a more reliable decision with a controlled error probability. The other agent, forced to decide simultaneously, faces a much higher probability of error---equivalent to what they would achieve by random guessing in the worst case. This competitive structure creates a ``first-to-decide'' advantage within collaborative information-gathering systems, where agents benefit from reaching confidence thresholds before their counterparts while maintaining acceptable accuracy levels.

\begin{remark}
Despite the agents behaving in a competing manner, the proposed framework is beneficial for agents in terms of stopping time. For a given probability of error, the stopping time for the process is less than or equal to the stopping time of the agents for every sample path (i.e., $\tau = \min(\tau^\A, \tau^\B)$) and in expectation. 
\hfill $\lrcorner$
\end{remark}

\begin{algorithm}
\caption{Sequential Binary Hypothesis Testing with Competing Agents}\label{Algorithm 1}
\begin{algorithmic}[1]
\Procedure{SBHT}{}
\State Each agent~$i\in\{\A,\B\}$ knows marginal distribution $\mathbb{P}(\theta = \theta^k \cap \{Y^i_t = y^i_t\}^n_{t=1})$, for $k=0,1$
\State $\mathtt{Stop} \gets 0$ and iteration $n \gets 1$
\State Desired probability of error $\beta \ll 0.5$
\While {$\mathtt{Stop} =0 $}
\State Agents collect observations $Y^i_n$, $i=\A,\B$
\State Perform local Bayesian update $\pi^{i}_n$, $i = \A, \B$
\State Transmit $\tilde\pi^i_n$ to $\Not i$, $i=\A,\B$, (\cref{Proposition 2})
\State Modify belief $\bar\pi^i_n$, $i=\A,\B$, (\cref{Proposition 3})
\If {$\exists i\in\{\A,\B\}$, $\exists k\in\{0,1\}$: $\pi^i_n(\theta^k)\leq \beta$,}
\State $\mathtt{Stop} \gets 1$
\State Decision policy $\mu^i_n(\bar\pi^i_n)$ given in \eqref{Eq: Decision Policy Stopping Agent}
\State Decision policy $\mu^{\Not i}_n(\bar\pi^i_n)$ given in \eqref{Eq: Decision Policy Alternate Agent}
\Else 
\State $n \gets n +1 $
\EndIf
\EndWhile
\EndProcedure
\end{algorithmic}
\end{algorithm}

\section{Numerical Example}

This section demonstrates the algorithm with an example. Suppose the common prior known to both agents be
\begin{align*}
    \mathbb{P}(\theta = \theta^1) =  \mathbb{P}(\theta = \theta^0) = \frac{1}{2}.
\end{align*}
The distributions of agents under the null ($\theta=\theta^0$) and the alternate ($\theta=\theta^1$) hypotheses are listed in \cref{Table 1}. The observations are considered to be independent in time.
We have
\begin{align*}
    D_{KL}( \mathbb{P} (Y^\A_n \mid \theta &= \theta^0) ,\mathbb{P} (Y^\A_n \mid \theta = \theta^1) ) =0.14 \\
    D_{KL}( \mathbb{P} (Y^\B_n \mid \theta &= \theta^0) ,\mathbb{P} (Y^\B_n \mid \theta = \theta^1) ) =0.496. 
\end{align*}
In \cref{fig:prob-error}, we plot the decay of the conditional probability of error for both agents across two different trials. We set $\beta=0.05$. Observe that, in trial 1 agent~$\A$ achieves the desired probability of error in 14 iterations, while agent~$\B$ achieves the same in only 10 iterations. In trial 2, agent~$\A$ takes 21 iterations while agent~$\B$ takes only 5 iterations. With this value of $\beta$, the average stopping time over 100 trials is as follows:
\begin{align*}
    \tau_{\text{avg}} = 6.95, \quad \tau^{\A}_{\text{avg}} = 19.70, \quad \tau^{\B}_{\text{avg}} = 7.95.
\end{align*}
With $\beta=0.01$, we have
\begin{align*}
    \tau_{\text{avg}} = 10.34, \quad \tau^{\A}_{\text{avg}} = 30.52, \quad \tau^{\B}_{\text{avg}} = 10.85.
\end{align*}
Thus, $\tau_{\text{avg}}$ for the system is closer to $\tau^{\B}_{\text{avg}}$, the average stopping time for agent~$\B$. This is because the higher KL divergence between the distributions under the null and alternative hypotheses enables agent~$\B$ to distinguish between the two hypotheses quickly. Among the $100$ trials, there were very few trials where agent $\A$ terminated quickly or won the competition. 
\begin{table}[!t]
\centering
\caption{Distributions $\mathbb P(Y_n^i \mid \theta)$ of the agents under different states of nature.}
\label{Table 1}
\begin{tabular}{|| c | c| c |c |c |c ||} 
 \hline
 Agent $\A$ & & & & & \\
 \hline
 $\theta = \theta^0$ & $0.1$ & $0.2$ & $0.1$ & $0.3$ & $0.3$ \\
 \hline
$\theta = \theta^1$ & $0.2$ & $0.15$ & $0.25$ & $0.2$ & $0.2$ \\
 \hline
 Agent $\B$ & & & & & \\
 \hline
 $\theta = \theta^0$ & $0.15$ & $0.25$ & $0.15$ & $0.25$ & $0.2$ \\
 \hline
$\theta = \theta^1$ & $0.4$ & $0.05$ & $0.35$ & $0.1$ & $0.1$ \\
 \hline
\end{tabular}
\end{table} 
\begin{figure}[!t]
    \centering
    \includegraphics[width = 0.45\textwidth]{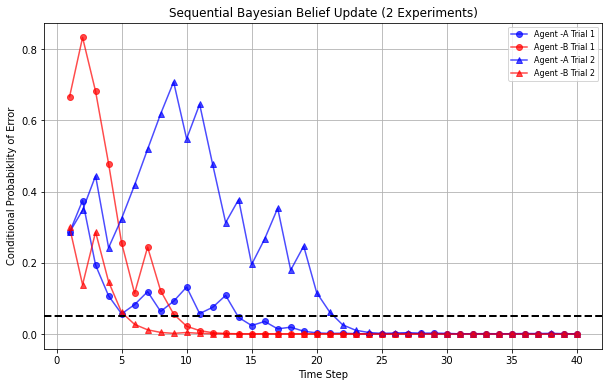}
    \caption{Decay of conditional probability of error with time}
    \label{fig:prob-error}
\end{figure}

\section{Conclusion and Future Work}

This paper investigated sequential hypothesis testing in competitive multi-agent environments under information asymmetry. Our analysis showed that when agents interact in settings with unreliable or strategically manipulated communication channels, conventional information-sharing protocols require fundamental reconsideration. Specifically, we showed that randomizing between truthful and inverted beliefs is the optimal signaling strategy to prevent the other agent from stopping first. This strategy effectively maximizes uncertainty while preserving information content. We also showed that agents achieve optimal performance by maintaining separate processing streams for their own observations versus information received from others. In particular, agents should update their beliefs based exclusively on their private observations until reaching their confidence thresholds. The information received from counterparts serves primarily as a mechanism to anticipate stopping decisions rather than as a direct input to belief formation. This result provides a practical approach to managing potentially manipulated information in distributed decision-making systems. Finally, we showed that the agent who reaches their confidence threshold first has an advantage, which creates an asymmetry with the first agent achieving bounded error probability while the second facing compromised accuracy.

Several promising directions exist for extending this research: Our current model considers a two-agent system with specific information structures. Extending the analysis to networks with multiple agents would enhance the applicability of our work to organizational and social systems. Secondly, investigating alternative observation models and information structures would strengthen the robustness of our results. This includes, for instance, scenarios with non-stationary observation distributions. Thirdly, developing adaptive learning mechanisms that enable agents to discover optimal strategies under uncertainty about their counterpart's behavior represents an important practical extension. This would involve incorporating online learning techniques to estimate the truthfulness parameters of other agents dynamically. Finally, relaxing the assumption of binary hypotheses to accommodate multiple possible states would generalize our results to scenarios with richer state spaces. This extension would increase the complexity of both the analysis and resulting strategies but would enable an application to settings where binary classifications are insufficient.

\bibliographystyle{ieeetr}
\bibliography{Biblio}
\end{document}